\documentclass[12pt]{article}
\usepackage{amsfonts}
\usepackage{amssymb}
\usepackage{mathrsfs}
\usepackage{srcltx}
\usepackage{amsmath}
\usepackage{amsthm}
\usepackage{amstext}
\usepackage{amsopn}
\usepackage{graphicx} 
\usepackage{epstopdf}
\newtheorem{theorem}{Theorem}[section]
\newtheorem{lemma}[theorem]{Lemma}

\newtheorem{corollary}[theorem]{Corollary}

\theoremstyle{definition}

\theoremstyle{remark}

\numberwithin{equation}{section}

\newcommand{\al}{\alpha}
\newcommand{\bt}{\beta}
\newcommand{\ba}{\begin{array} }
\newcommand{\ea}{\end{array} }
\newcommand{\be}{\begin{equation} }
\newcommand{\ee}{\end{equation} }
\newcommand{\baa}{\begin{align} }
\newcommand{\eaa}{\end{align} }
\newcommand{\da}{\delta}

\newcommand{\la}{\lambda}
\newcommand{\f}{\displaystyle\frac}
\newcommand{\il}{\displaystyle\int}

\newcommand{\vp}{\varphi}

\newcommand{\ga}{\gamma}

\newcommand{\ra}{\rightarrow}

\newcommand{\og}{\omega}

\newcommand{\sg}{\sigma}

\newcommand{\R}{{\mathbb{R}}}

\newcommand{\pa}{\partial}

\newcommand{\tu}[1]{\textup{#1}}

\newcommand{\iy}{\infty}

\DeclareMathOperator{\ran}{range}\DeclareMathOperator{\cod}{codim}
\DeclareMathOperator{\spa}{span}
\begin{document}
\date{}
\title{\bf\large{ Model of Algal Growth Depending on Nutrients and Inorganic Carbon in a Poorly Mixed Water Column}\thanks{Partially supported by NSFC-11971088, NSFC-11901140, NSFHLJ-LH2019A022, and NSF-DMS-1853598}}
\author{Jimin Zhang$^a$, Junping Shi$^b$\thanks{Corresponding
author. E-mail address: jxshix@wm.edu}, Xiaoyuan Chang$^c$}

\maketitle

\begin{center}
\begin{minipage}{12cm}
\begin{description}
\item \small $a$.~School of Mathematical Sciences, Heilongjiang
University,  Harbin, Heilongjiang, 150080, P.R. China

\item \small $b$.~Department of Mathematics, William
\& Mary, Williamsburg, VA, 23187-8795, USA

\item \small $c$.~School of Science, Harbin University of
Science and Technology,  Harbin, Heilongjiang,
150080, P.R. China

\end{description}
\end{minipage}
\end{center}

\vskip 0.2in

\begin{abstract}
In this paper, we establish a reaction-diffusion-advection partial differential equation model to describe the growth of algae depending on both nutrients and inorganic carbon in a poorly mixed water column. Nutrients from the water bottom and inorganic carbon from the water surface form an asymmetric resource supply mechanism on the algal growth. The existence and stability of semi-trivial steady state and coexistence steady state of the model are proved, and a threshold condition for the regime shift from extinction to survival of algae is established. The influence of environmental parameters on the vertical distribution of algae is investigated in the water column. It is shown that the vertical distribution of algae can exhibit many different profiles under the joint limitation of nutrients and inorganic carbon.
\end{abstract}



\section{Introduction}
\noindent

Algae are the basis of earth food web and preserve the balance of the global aquatic ecosystem, and they are adaptable and widely distributed in rivers, lakes and oceans. Photosynthesis of algae consumes inorganic carbon to produce organic matter and release oxygen by using light energy. In this procedure, inorganic carbon and light are involved in energy flow and material cycle of the aquatic ecosystem. Nutrient elements, such as phosphorus and nitrogen, are key factors for algal growth and important indicators of water eutrophication. Therefore the growth of algae is supported and restricted  by three essential resources: light, nutrients (i.e. nitrogen and phosphorus) and inorganic carbon (i.e. dissolved CO$_2$, carbonic acid and bicarbonate). Understanding the algae growth in aquatic environment is of fundamental importance to the ecosystem studies. 

Mathematical models have been constructed to study the growth mechanism of algae and its dependence on algae and nutrients, or light, or both of them.
Three different situations have been studied and discussed. First, algae compete only for
nutrients in oligotrophic aquatic ecosystems with ample supply of light \cite{HusWangZhao2013,Nie2015,ShiZhangZhang2019,WangHsuZhao2015,ZhangShiChang2018}; Second, algae compete only for light in eutrophic aquatic ecosystems \cite{DuHsu2010,DuHsuLou2015,Du2011,HusLou2010,JiangLamLouWang2019,MeiZhang2012,PengZhao2010,Zagaris2011}; Third, algae compete for both light and nutrients simultaneously in some
aquatic ecosystems \cite{DuHsu2008a,DuHsu2008b,Huismanv2006,Jager2010,Klausmeier2001,MeiZhang2012DCDSB,Ryabov2010,Vasconcelos2016,Yoshiyama2009,Yoshiyama2002}.

The connection between algae and inorganic carbon is more complicated with a variety of biological mechanisms. Carbon dioxide enters the water by exchange at the water surface, and reacts with water molecules to form carbonic acid and bicarbonate. Algae could take up
dissolve CO$_2$, carbonic acid and bicarbonate by photosynthesis.  An ODE model was constructed in \cite{Van2011} to describe competition for dissolved inorganic carbon in dense algal blooms in a completely well-mixed water column. In \cite{Nie2016} and \cite{HsuLamWang2017}, the authors established PDE models to deal with the interaction between algae and inorganic carbon in a poorly mixed habitat.

The main purpose of this paper is to establish a mathematical model to describe the interaction of algae, nutrients and inorganic carbon in a poorly mixed water column with ample supply of light. The growth of algae only depends on nutrients from the bottom and inorganic carbon from the surface (see Figure \ref{figureshiytu}). The increase of the algal biomass on the water surface inhibits the algal growth on the deep layer since limited inorganic carbon from the surface decreases its supply to the deep layer. On the other hand, an increase of the algal biomass on the water bottom also suppresses the algal growth in the surface layer since limited nutrients from the water bottom decreases its supply to the surface layer. This forms an asymmetric competition for the algae to gain the two resources. It is important and of interest to explore the effect of this asymmetric resource competition for nutrients and inorganic carbon on the algal growth and vertical distribution.

The vertical distribution of algae is highly heterogeneous and exhibits the most prominent vertical aggregation phenomena. For example, algae usually float on the water surface during the daytime and sink to the water bottom at night.
The spatial heterogeneity of algae is generally related to the uneven distribution of essential resources. Previous studies have shown that algae have complex vertical distribution with supply of light and nutrients in poorly mixed water columns \cite{Klausmeier2001,Ryabov2010,Yoshiyama2009,Yoshiyama2002}. In the present paper,
we will reveal that algae also exhibit vertically spatial heterogeneity and vertical aggregation phenomena under the asymmetric resource supply mechanism of nutrients and inorganic carbon, which has not been considered in previous studies.

\begin{figure}[htb]
\centerline{\includegraphics[scale=0.8]{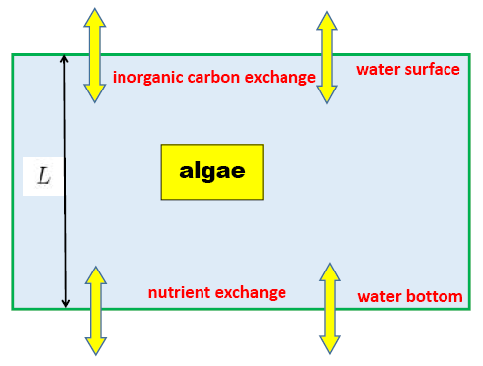}}
\caption{{\footnotesize  Interactions of
algae, nutrients and inorganic carbon in a water column.}}\label{figureshiytu}
\end{figure}

The rest of the paper is organized as follows. In  Section \ref{sectionmodel}, we derive a reaction-diffusion-advection PDE model to describe the growth of algae depending on both nutrients and inorganic carbon in a poorly mixed water column. We then investigate the basic dynamics of this model including the existence and stability of steady states in Section  \ref{sectionsteadystate}. In Section \ref{sectionvertical}, we consider the influence of environmental parameters on the vertical distribution of algae, and explain what mechanisms drive these vertical distributions. Finally, we summarize our findings and state some questions for future study in Section \ref{sectiondiscussion}. Throughout the paper, numerical simulations under reasonable parameter values from literature are presented to illustrate or complement our mathematical findings.

\section{Model construction}\label{sectionmodel}
\noindent

In this section, we establish a mathematical model to describe the interactions of
algae, nutrients  and inorganic carbon in a poorly mixed water column with ample supply of light (see Figure \ref{figureshiytu}). Let $z$ denote the depth coordinate of the water column. We assume that $z=0$ is the surface of the water column,  and $z=L$ is the bottom of the water column. Three partial differential equations are established below to describe the dynamics of biomass density of algae $A$, concentration of dissolved nutrients $N$ (i.e. nitrogen and phosphorus), and concentration of dissolved inorganic carbon $C$ (i.e. dissolved CO$_2$, carbonic acid and bicarbonate). All the variables and parameters of the model and their biological significance are listed in Table \ref{tabletparameters}.
\begin{table}[h]
{\begin{center} \caption[]{ Variables and parameters
of model \eqref{fullequation} with biological meanings.
}\label{tabletparameters}\scriptsize
\begin{tabular}{p{1.2cm}p{5.6cm}p{1.2cm}p{5.6cm}}\hline
{\bf  Symbol} &{\bf  Meaning}&{\bf  Symbol}&{\bf Meaning}\\
\hline$t$&Time&$z$&Depth\\
$A$&Biomass density of algae&$N$& Concentration of dissolved nutrients\\
$C$& Concentration of dissolved
inorganic carbon&$D_a$&
Vertical turbulent diffusivity of algae\\
$D_n$&
Vertical turbulent diffusivity of dissolved nutrients&$D_c$&
Vertical turbulent diffusivity of dissolved
inorganic carbon\\
$s$&
Sinking or buoyant velocity of algae&$r$&
Maximum specific production rate of algae\\
$m$& Loss rate of algae&
$l$&Respiration rate of algae\\
$Q$& Carbon
quota of algae&$\ga_n,\ga_c$&Half saturation constant for nutrient-limited and inorganic carbon-limited production of algae respectively\\
$c_n$&Phosphorus to carbon quota of algae&$p_n$& Proportion of nutrient in algal losses that is recycled\\
$p_c$& The proportion of carbon dioxide released by algae respiration in the
water column&$\al$& Nutrient exchange rate at the bottom of the water column\\
$\bt$&  CO$_2$ gas exchange rate between air and water at the surface of the water column &$N^0$&Concentration of dissolved nutrients at the bottom of the water column\\
$C^0$&Concentration of dissolved inorganic carbon at the surface of the water column
&$L$&  Depth of the water column\\
\hline
\end{tabular}
\end{center}}
\end{table}

Let $A(z,t)$  denote the biomass density of algae at depth $z\in[0,L]$ and time $t$. The growth rate of algae has two limiting factors:
concentration of
dissolved nutrients $N(z,t)$ and concentration of dissolved inorganic carbon $C(z,t)$.
It takes a form of multiplication of two Monod type equations:
\[
f(N)g(C)=\f{N}{\ga_n+N}\f{C}{\ga_c+C}.
\]
Biomass density of algae is lost at density-independent
rate $m$, caused by
death and grazing, and respiration $l/Q$,  where $l$ is the respiration rate
and $Q$ is the carbon quota of algae.  On the other hand, algae may move up or down by turbulence with a depth independent diffusion coefficient
$D_a$. In addition, it also has an active movement to sinking or buoyant with speed $s$ because of photokinesis or seeking resources.
Combining these assumptions gives the following reaction-diffusion-advection equation
with no-flux boundary condition:
\begin{equation}\label{alageequation}
\begin{split}
\f{\pa A(z,t)}{\pa t}&=\text{turbulent diffusion}-\text{sinking} (\text{or buoyant})+\text{growth}-\text{loss},\\
&= D_a\f{\pa^2A}{\pa z^2}-s\f{\pa A}{\pa
z}+rf(N)g(C)A-\left(m+\f{l}{Q}\right)A,~0<z<L,\\
D_aA_z(0,t)&-sA(0,t)=0,~D_aA_z(L,t)-sA(L,t)=0.
\end{split}
\end{equation}

The function $N(z,t)$ describes the concentration of dissolved nutrients in the water column at depth $z\in[0,L]$ and time $t$. The nutrients are supplied from the bottom of the water column with a fixed concentration $N^0$ and nutrient exchange rate $\al$.
The nutrient transport is also governed by passive
movement due to turbulence with a diffusion coefficient $D_n$.  $N(z,t)$ decreases as
it is consumed by algae with consumption rate $c_nrf(N)g(C)A$, and increases  as a result of  recycling from loss of algal biomass with
proportion $p_n\in[0,1]$. The dynamics of $N(z,t)$ is given by
\begin{equation}\label{nutrientequation}
\begin{split}
\f{\pa N(z,t)}{\pa t}&=\text{turbulent diffusion}+\text{recycling}-\text{consumption}\\
&=D_n\f{\pa^2N}{\pa z^2}+p_nc_nmA
-c_nrf(N)g(C)A,~0<z<L,\\
N_z(0,t)&=0,~D_nN_z(L,t)=\al(N^0-N(L,t))~\text{(nutrients
exchange)}.
\end{split}
\end{equation}

Let $C(z,t)$ be the concentration of dissolved
inorganic carbon in the water column at depth $z\in[0,L]$ and time $t$.
Inorganic carbon in the water column mainly comes from the atmosphere, and a very small part comes from the sediment. In order to study the asymmetry of resource supply in the present paper, we assume that inorganic carbon is only supplied from the atmosphere at the surface of the water column with a fixed concentration $C^0$. The change of dissolved  inorganic carbon depends on turbulent diffusion with a diffusion coefficient $D_c$, consumption by algae, recycling from respiration of algae with
proportion $p_c\in[0,1]$, and CO$_2$ gas exchange between air and water with exchange rate $\bt$. The dynamics of $C(z,t)$ is described as
\begin{equation}\label{carbonequation}
\begin{split}
\f{\pa C(z,t)}{\pa t}&=\text{turbulent diffusion}+\text{recycling}-\text{consumption}\\
&=D_c\f{\pa^2C}{\pa z^2}+\f{p_cl}{Q}A
-rf(N)g(C)A,~0<z<L,\\
D_cC_z(0,t)&=\bt(C(0,t)-C^0)~\text{(CO$_2$
exchange)},~C_z(L,t)=0.
\end{split}
\end{equation}

Combining all equations \eqref{alageequation}-\eqref{carbonequation}, we have the following full system of algae-nutrient-inorganic carbon model:
\begin{equation}\label{fullequation}
\begin{cases}
\f{\pa A}{\pa t}= D_a\f{\pa^2A}{\pa z^2}-s\f{\pa A}{\pa
z}+rf(N)g(C)A-\left(m+\f{l}{Q}\right)A,&0<z<L,~t>0,\\
\f{\pa N}{\pa t}=D_n\f{\pa^2N}{\pa z^2}+p_nc_nmA
-c_nrf(N)g(C)A,&0<z<L,~t>0,\\
\f{\pa C}{\pa t}=D_c\f{\pa^2C}{\pa z^2}+\f{p_cl}{Q}A-rf(N)g(C)A,&0<z<L,~t>0,\\
D_aA_z(0,t)-sA(0,t)=0,~D_aA_z(L,t)-sA(L,t)=0,&t>0,\\
N_z(0,t)=0,~D_nN_z(L,t)=\al(N^0-N(L,t)),&t>0,\\
D_cC_z(0,t)=\bt(C(0,t)-C^0),~C_z(L,t)=0,&t>0.
\end{cases}
\end{equation}

Due to the biological meaning of variables in \eqref{fullequation}, we will deal with the solutions of \eqref{fullequation} with nonnegative initial values, i.e.
$$
A(z,0)=A_0(z)\geq\not\equiv0,~N(z,0)=N_0(z)\geq\not\equiv0,~C(z,0)=C_0(z)\geq\not\equiv0.
$$
We also assume that $s\in\R$, $p_n,p_c\in[0,1]$ and all remaining parameters are positive constants unless explicitly stated otherwise.

\section{Existence and stability of steady states}\label{sectionsteadystate}
\noindent

In this section, we explore the existence and stability of steady state solutions of \eqref{fullequation}. The steady state solutions of \eqref{fullequation} satisfy
\begin{equation}\label{steadystateequation}
\begin{cases}
D_aA''(z)-sA'(z)+rf(N(z))g(C(z))A(z)-\left(m+\f{l}{Q}\right)A(z)=0,&0<z<L,\\
D_nN''(z)+p_nc_nmA(z)
-c_nrf(N(z))g(C(z))A(z)=0,&0<z<L,\\
D_cC''(z)+\f{p_cl}{Q}A(z)-rf(N(z))g(C(z))A(z)=0,&0<z<L,\\
D_aA'(0)-sA(0)=D_aA'(L)-sA(L)=0,&\\
N'(0)=0,~D_nN'(L)=\al(N^0-N(L)),&\\
D_cC'(0)=\bt(C(0)-C^0),~C'(L)=0.&
\end{cases}
\end{equation}

To establish the local stability of the steady state solutions, we consider
the following eigenvalue problem
\begin{equation}\label{eqcharacterequation}
\begin{cases}
\la\vp(z)=D_a\vp''(z)-s\vp'(z)+\left[rf(\bar{n})g(\bar{c})-(m+(l/Q))\right]\vp(z)&\\
~~~~~~~~~~~+g(\bar{c})\f{r\ga_n\bar{a}}{(\ga_n+\bar{n})^2}\psi(z)+f(\bar{n})
\f{r\ga_c\bar{a}}{(\ga_c+\bar{c})^2}\phi(z),&0<z<L,\\
\la\psi(z)=D_n\psi''(z)+c_n(p_nm-rf(\bar{n})g(\bar{c}))\vp(z)\\
~~~~~~~~~~~-g(\bar{c})\f{c_nr\ga_n\bar{a}}{(\ga_n+\bar{n})^2}\psi(z)-f(\bar{n})
\f{c_nr\ga_c\bar{a}}{(\ga_c+\bar{c})^2}\phi(z),
&0<z<L,\\
\la\phi(z)=D_c\phi''(z)+((p_cl/Q)-rf(\bar{n})g(\bar{c}))\vp(z)\\
~~~~~~~~~~~-g(\bar{c})\f{r\ga_n\bar{a}}{(\ga_n+\bar{n})^2}\psi(z)
-f(\bar{n})\f{r\ga_c\bar{a}}{(\ga_c+\bar{c})^2}\phi(z),
&0<z<L,\\
D_a\vp'(0)-s\vp(0)=D_a\vp'(L)-s\vp(L)=0, &\\
\psi'(0)=0,~D_n\psi'(L)+\al\psi(L))=0, &\\
-D_c\phi'(0)+\bt\phi(0)=0,~\phi'(L)=0,&
\end{cases}
\end{equation}
where $(\bar{a}(z),\bar{n}(z),\bar{c}(z))$ is a steady state solution
of \eqref{fullequation}. When all
eigenvalues of \eqref{eqcharacterequation} have negative real
part,  $(\bar{a}(z),\bar{n}(z),\bar{c}(z))$ is locally
asymptotically stable, otherwise it is unstable with at least one
eigenvalue having positive real part.

\subsection{Nutrient-inorganic carbon-only semi-trivial steady state}

A nutrient-inorganic carbon-only semi-trivial steady state $E_1:(0,N_1(z),C_1(z))$ satisfies
\begin{equation}\label{NutrientsCarbonequation}
\begin{cases}
N''(z)=0,~N'(0)=0,~D_nN'(L)=\al(N^0-N(L)),&0<z<L,\\
C''(z)=0,~D_cC'(0)=\bt(C(0)-C^0),~C'(L)=0,&0<z<L.
\end{cases}
\end{equation}
We have the following results regarding the existence, uniqueness and local stability of
$E_1$.
\begin{theorem}\label{thNClocal}
\begin{itemize}
\item[\tu{(i)}] The system \eqref{fullequation} has a unique nutrient-inorganic carbon-only steady state solution $E_1\equiv(0,N^0,C^0)$;
\item[\tu{(ii)}]  If
\begin{equation}\label{eqNCcondition}
m+\f{l}{Q}>rf(N^0)g(C^0),
\end{equation}
then $E_1$ is locally asymptotically stable with respect to \eqref{fullequation}, and if
\begin{equation}\label{eqNCconditionunstable}
m+\f{l}{Q}<rf(N^0)g(C^0),
\end{equation}
then $E_1$ is unstable.
\end{itemize}
\end{theorem}
\begin{proof}
It follows from \eqref{NutrientsCarbonequation} that (i) holds.
From \eqref{eqcharacterequation}, the stability of $E_1$ is determined
by the eigenvalue problem
\begin{subequations}\label{eqNCequation}
\begin{align}
&\la\vp(z)=D_a\vp''(z)-s\vp'(z)+\left[rf(N^0)g(C^0)-(m+(l/Q))\right]\vp(z),~0<z<L,
\label{eqNCequationa}\\
&\la\psi(z)=D_n\psi''(z)+c_n(p_nm-rf(N^0)g(C^0))\vp(z),
~0<z<L,\label{eqNCequationb}\\
&\la\phi(z)=D_c\phi''(z)+((p_cl/Q)-rf(N^0)g(C^0))\vp(z),
~0<z<L,\label{eqNCequationc}\\
&D_a\vp'(0)-s\vp(0)=D_a\vp'(L)-s\vp(L)=0, \label{eqNCequationd}\\
&\psi'(0)=0,~D_n\psi'(L)+\al\psi(L)=0,\label{eqNCequatione}\\
&-D_c\phi'(0)+\bt\phi(0)=0,~\phi'(L)=0.\label{eqNCequationf}
\end{align}
\end{subequations}
Note that the linearized system \eqref{eqNCequation} is partially decoupled. We consider the following two cases:  $\vp\neq0$ and $\vp\equiv0$.

Case 1: $\vp\neq0$. In this case, the stability of $E_1$ is determined by \eqref{eqNCequationa} and its boundary condition \eqref{eqNCequationd}. Let $\vp=e^{(s/D_a)z}\eta$. Then \eqref{eqNCequationa} translates into
\begin{equation}\label{eqNCequationaa}
\begin{cases}
\la\eta(z)=D_a\eta''(z)+s\eta'(z)+\left[rf(N^0)g(C^0)-(m+(l/Q))\right]\eta(z)&0<z<L,\\
\eta'(0)=\eta'(L)=0.&
\end{cases}
\end{equation}
It is not difficult to show that the dominant eigenvalue of \eqref{eqNCequationaa}
is $rf(N^0)g(C^0)-(m+(l/Q))$ and the corresponding
eigenfunction is $\eta(z)\equiv1$. This implies that $\la_1 = rf(N^0)g(C^0)-(m+(l/Q))$ is also an eigenvalue of \eqref{eqNCequation},
the corresponding eigenfunction is $\vp(z) = e^{(s/Du)z}$, and $(\psi,\phi)$ can be solved from
\eqref{eqNCequation}.

Case 2: $\vp\equiv0$. In this case, \eqref{eqNCequationb} and \eqref{eqNCequationc} with their boundary conditions \eqref{eqNCequatione} and \eqref{eqNCequationf} reduce to
\begin{equation}\label{eqNCequationbb}
\la\psi(z)=D_n\psi''(z),~0<z<L,~\psi'(0)=0,~D_n\psi'(L)+\al\psi(L)=0
\end{equation}
and
\begin{equation}\label{eqNCequationcc}
\la\phi(z)=D_c\phi''(z),~0<z<L,~-D_c\phi'(0)+\bt\phi(0)=0,~\phi'(L)=0.
\end{equation}
The eigenvalues of \eqref{eqNCequationbb} must be negative. In fact, if $\la>0$ in \eqref{eqNCequationbb}, then we have $\psi(z)=\cosh(\sigma z)$ for $\sigma=\sqrt{\la/D_n}$ since $\psi'(0)=0$. But from $D_n\psi'(L)+\al\psi(L)$=0, we obtain $D_n\sinh(\sigma L)=-\al\cosh(\sigma L)$, which is a contradiction. It is also easy to see $\la=0$ cannot be an eigenvalue of \eqref{eqNCequationbb}. If $\la<0$, then from $\psi'(0)=0$ we have $\psi(z)=\cos(\sg z)$ for $\sg=\sqrt{-\la/D_n}$. It follows from $D_n\psi'(L)+\al\psi(L)=0$ that $\tan \sg L=\al/(\sg D_n)$. Then the dominant eigenvalue of \eqref{eqNCequationbb} is $\la_1=-D_n\sg_1^2$, where $\sg_1$ is the smallest positive root of $\tan \sg L=\al/(\sg D_n)$. Similarly, the eigenvalues of \eqref{eqNCequationcc} are also negative.

Summarizing above discussions, we conclude that if \eqref{eqNCcondition} holds, then $\la_1<0$ and $E_1$ is locally asymptotically stable. On the other hand, if \eqref{eqNCconditionunstable}
holds, then $\la_1>0$ and $E_1$ is unstable.
\end{proof}

The condition \eqref{eqNCcondition} in Theorem \ref{thNClocal} shows that a large algal loss rate $m$ or respiration rate $l$ leads to the extinction of algae population. We next prove that if the composition of algal loss rate and respiratory rate exceeds its maximum specific production rate, then the extinction of algae population is global for all initial conditions.

\begin{theorem}\label{thNCglobal}
If
\begin{equation}\label{eqNCglobalcondition}
m+\f{l}{Q}>r,
\end{equation}
then $E_1$ is globally asymptotically
stable with respect to \eqref{fullequation} for any nonnegative initial value.
\end{theorem}
\begin{proof}
From the first equation of \eqref{fullequation} and \eqref{eqNCglobalcondition}, we have
\begin{equation}\label{eqcongpingping}
\f{\pa A}{\pa t}\leq D_a\f{\pa^2A}{\pa z^2}-s\f{\pa A}{\pa
z}+\left(r-\left(m+\f{l}{Q}\right)\right)A,~0<z<L,~t>0,
\end{equation}
which implies that $A(z,t)$ converges to $0$ uniformly for $z\in[0,L]$ as
$t\ra\iy$ by the comparison theorem of parabolic equations. It follows from the theory of asymptotic autonomous systems \cite{Mischaikow1995} that the dynamics of \eqref{fullequation} reduces to the one of limiting system
\begin{equation}\label{eqglobal}
\begin{cases}
\f{\pa N}{\pa t}=D_n\f{\pa^2N}{\pa z^2},
&0<z<L,~t>0,\\
\f{\pa C}{\pa t}=D_c\f{\pa^2C}{\pa z^2},&0<z<L,~t>0,\\
N_z(0,t)=0,~D_nN_z(L,t)=\al(N^0-N(L,t)),&t>0,\\
D_cC_z(0,t)=\bt(C(0,t)-C^0),~C_z(L,t)=0,&t>0.
\end{cases}
\end{equation}
To obtain our conclusions, we define the following Lyapunov functional $V:C([0,L])\times C([0,L])\ra\R$ by
$$
V(N,C)=\f{1}{2}\il_0^L(N(z)-N^0)^2dz+\f{1}{2}\il_0^L(C(z)-C^0)^2dz.
$$
Let $(N(z,t),C(z,t))$ be an arbitrary solution of \eqref{eqglobal} with nonnegative initial values. Then
\begin{align*}
\f{dV(N(z,t),C(z,t))}{dt}&=\il_0^L(N(z,t)-N^0)\f{\pa N}{\pa t}dz+\il_0^L(C(z,t)-C^0)\f{\pa C}{\pa t}dz\\
&=D_n\il_0^L(N(z,t)-N^0)\f{\pa N^2}{\pa z^2}dz+D_c\il_0^L(C(z,t)-C^0)\f{\pa C^2}{\pa z^2}dz\\
&=D_n\f{\pa N}{\pa z}(N(z,t)-N^0)\Big |_0^L-D_n\il_0^L\left(\f{\pa N}{\pa z}\right)^2dz\\
&\quad+D_c\f{\pa C}{\pa z}(C(z,t)-C^0)\Big |_0^L-D_c\il_0^L\left(\f{\pa C}{\pa z}\right)^2dz\\
&=-\al(N(z,t)-N^0)^2-D_n\il_0^L\left(\f{\pa N}{\pa z}\right)^2dz\\
&\quad-\bt(C(z,t)-C^0)^2-D_c\il_0^L\left(\f{\pa C}{\pa z}\right)^2dz\leq0.
\end{align*}
Note that $dV(\cdot)/dt=0$ holds if and only if
$N(z,t)\equiv N^0$ and $C(z,t)\equiv C^0$. From the LaSalle's Invariance
Principle \cite[Theorem 4.3.4]{Henry2006}, we conclude that $(N(z,t),C(z,t))$ converges to $(N^0,C^0)$ uniformly for $z\in[0,L]$ as $t\ra\iy$. This means that $E_1$ is globally asymptotically
stable with respect to \eqref{fullequation} for any nonnegative initial value.
\end{proof}

We next show that if the recycling of nutrients and carbon dioxide is not considered, then the local stability of $E_1$ implies  that the global stability of $E_1$.
\begin{corollary}
If $p_n=p_c=0$ and \eqref{eqNCcondition} hold, then $E_1$ is globally asymptotically
stable with respect to \eqref{fullequation} for any nonnegative initial value.
\end{corollary}
\begin{proof}
From \eqref{fullequation}, we have
\begin{equation*}
\begin{cases}
\f{\pa N}{\pa t}\leq D_n\f{\pa N^2}{\pa z^2},&0<z<L,~t>0,\\
\f{\pa C}{\pa t}\leq D_c\f{\pa^2C}{\pa z^2},&0<z<L,~t>0,\\
N_z(0,t)=0,~D_nN_z(L,t)=\al(N^0-N(L,t)),&t>0,\\
D_cC_z(0,t)=\bt(C(0,t)-C^0),~C_z(L,t)=0,&t>0.
\end{cases}
\end{equation*}
By using the comparison principle of parabolic equations,
we get
\[
N(z,t)\leq N^0,~C(z,t)\leq C^0,~0\leq z\leq L,~t>0.
\]
It follows from \eqref{eqcongpingping} that
\[
\f{\pa A}{\pa t}\leq D_a\f{\pa^2A}{\pa z^2}-s\f{\pa A}{\pa
z}+\left(rf(N^0)g(C^0)-\left(m+\f{l}{Q}\right)\right)A,~0<z<L,~t>0.
\]
Then $A(z,t)$ converges to $0$ uniformly for $z\in[0,L]$ as
$t\ra\iy$  if \eqref{eqNCcondition} holds. The rest of the proof is similar to the proof of Theorem \ref{thNCglobal}.
\end{proof}

\subsection{Coexistence steady state}

The main purpose of this subsection is to investigate the coexistence of algae, nutrients and inorganic carbon in a water column by using bifurcation theory. A coexistence steady state solution $E_2:(A(z),N(z),C(z))$ satisfies \eqref{steadystateequation} and each of $A(z),N(z)$ and $C(z)$ is positive. Define
\begin{equation}\label{mpoint}
m_*=rf(N^0)g(C^0)-\f{l}{Q}>0.
\end{equation}

We first consider the local bifurcation of the coexistence steady state $E_2$ from nutrients-inorganic carbon-only semi-trivial steady state $E_1$ at $m=m_*$. In order to obtain our results, we define function spaces
\begin{equation*}
\begin{split}
X_1&:=\{u\in C^2[0,L]:D_au'(0)-su(0)=D_au'(L)-su(L)=0\},\\
X_2&:=\{u\in C^2[0,L]:u'(0)=0\},~
X_3:=\{u\in C^2[0,L]:u'(L)=0\},~
Y:=C[0,L],
\end{split}
\end{equation*}
and a nonlinear mapping
$F:\R^+\times X_1\times X_2\times X_3\ra Y\times Y\times Y\times \R\times\R$ by
\begin{equation*}
F(m,A(z),N(z),C(z))=\begin{pmatrix}
D_aA''(z)-sA'(z)+rf(N)g(C)A-(m+l/Q)A\\
D_nN''(z)+p_nc_nmA
-c_nrf(N)g(C)A\\
D_cC''(z)+(p_cl/Q)A-rf(N)g(C)A\\
D_nN'(L)-\al(N^0-N(L))\\
D_cC'(0)-\bt(C(0)-C^0)
\end{pmatrix}.
\end{equation*}
For any $(\vp,\psi,\phi)\in X_1\times X_2\times X_3$, the Fr\'{e}chet derivatives of $F$ at $(m,A,N,C)$ are
\begin{equation}\label{eqFrechetunc}
\begin{split}
&F_{(A,N,C)}(m,A(z),N(z),C(z))[\varphi,\psi,\phi]=
\\&\begin{pmatrix}
D_a\vp''-s\vp'+\left[rf(N)g(C)-\left(m+\f{l}{Q}\right)\right]\vp
+\f{r\ga_nA g(C)}{(\ga_n+N)^2}\psi+
\f{r\ga_cA f(N)}{(\ga_c+C)^2}\phi\\
D_n\psi''+c_n(p_nm-rf(N)g(C))\vp
-\f{c_nr\ga_nA g(C)}{(\ga_n+N)^2}\psi-
\f{c_nr\ga_cA f(N)}{(\ga_c+C)^2}\phi\\
D_c\phi''+\left(\f{p_cl}{Q}-rf(N)g(C)\right)\vp-\f{r\ga_nA g(C)}{(\ga_n+N)^2}\psi-
\f{r\ga_cA f(N)}{(\ga_c+C)^2}\phi\\
D_n\psi'(L)+\al\psi(L)\\
D_c\phi'(0)-\bt\phi(0)
\end{pmatrix}.
\end{split}
\end{equation}
For the convenience of the following discussion, we denote
\begin{equation}\label{eqbiaodashi}
\begin{split}
\bar{\vp}(z)&=e^{(s/D_a)z},~
\bar{\psi}(z)=c_1e^{(s/D_a)z}+c_2z+c_3,\\
\bar{\phi}(z)&=\f{D_a^2m_*}{D_cs^2}e^{(s/D_a)z}
+c_4z+c_5,~0<z<L,
\end{split}
\end{equation}
where
\begin{align*}
c_1&=\f{D_a^2(l/Q+m_*(1-p_n))}{D_n s^2},~c_2=-\f{D_a(l/Q+m_*(1-p_n))}{D_ns},\\
c_3&=-\left(\f{D_a(l/Q+m_*(1-p_n))}{D_ns\al}\right)\left(\f{(D_ns+D_a\al)e^{(s/D_a)L}}{s}+
\al L+D_n\right),\\
c_4&=-\f{D_a(m_*+(1-p_c)l/Q))}{D_cs}e^{(s/D_a)L},~c_5=-\f{D_a(m_*+(1-p_c)l/Q)}{s\bt}
\left(e^{(s/D_a)L}-1\right)-\f{D_a^2m_*}{D_cs^2}.
\end{align*}

\begin{theorem}\label{theoremlocal}
Assume that \eqref{mpoint} holds. Then
\begin{itemize}
\item[\tu{(i)}] The point $(m_*,0,N^0,C^0)$ is a bifurcation point of the
coexistence steady state solutions of \eqref{fullequation}.
Moreover, near $(m_*,0,N^0,C^0)$, there exists a positive constant $\delta>0$ such that the bifurcating coexistence steady state solutions near $(m_*,0,N^0,C^0)$ are on a smooth curve
$\Gamma_1=\{E_2(\tau)=(m(\tau),A(\tau,z),N(\tau,z),C(\tau,z)):0<\tau<\delta\}$ 
with
\begin{equation}\label{eqqncbiaodashiaaa}
\begin{cases}
A(\tau,z)=\tau\bar{\vp}(z)+o(\tau),\\
N(\tau,z)=N^0+\tau\bar{\psi}(z)+o(\tau),\\
C(\tau,z)=C^0+\tau\bar{\phi}(z)+o(\tau)
\end{cases}
\end{equation}
and $m'(0)<0$;

\item[\tu{(ii)}] $m_*$ is the unique bifurcation value for bifurcation of
 positive steady state solutions from the line of semi-trivial solutions  $\Gamma_0=\{(m,0,N^0,C^0):m>0\}$;

\item[\tu{(iii)}] For $\tau\in (0,\delta)$, the bifurcating solution $E_2(\tau)=\left(m(\tau), A(\tau,z), N(\tau,z), C(\tau,z)\right)$ is locally asymptotically stable with respect to \eqref{fullequation}.
 \end{itemize}
\end{theorem}
\begin{proof}
It is easy to see that $F(m,0,N^0,C^0)=0$ and
\begin{equation}\label{eqzhubajieaa}
\begin{split}
F_{(U,N,C)}(m_*,0,N^0,C^0)[\varphi,\psi,\phi]=
\begin{pmatrix}
D_a\vp''(z)-s\vp'(z)\\
D_n\psi''(z)+((p_n-1)m_*-l/Q)\vp\\
D_c\phi''(z)-(m_*+(1-p_c)l/Q)\vp\\
D_n\psi'(L)+\al\psi(L)\\
D_c\phi'(0)-\bt\phi(0)
\end{pmatrix}.
\end{split}
\end{equation}
Let $H:=F_{(U,N,C)}(m_*,0,N^0,C^0)$. The remaining proof of Theorem \ref{theoremlocal} is divided into the following several steps.

(i) We first show that $\dim\ker H=1$. If $(\bar{\vp}(z),\bar{\psi}(z),\bar{\phi}(z))\in \ker H$, then
\begin{subequations}\label{eqker}
\begin{align}
&D_a\bar{\vp}''(z)-s\bar{\vp}'(z)=0,~0<z<L,\label{eqkera}\\
&D_n\bar{\psi}''(z)+((p_n-1)m_*-l/Q)\bar{\vp}(z)=0,~0<z<L,\label{eqkerb}\\
&D_c\bar{\phi}''(z)-(m_*+(1-p_c)l/Q)\bar{\vp}(z)=0,~0<z<L,\label{eqkerc}\\
&D_n\bar{\psi}'(L)+\al\bar{\psi}(L)=0,\label{eqkerd}\\
&D_c\bar{\phi}'(0)-\bt\bar{\phi}(0)=0.\label{eqkere}
\end{align}
\end{subequations}
It follows from \eqref{eqkera} and $\bar{\vp}\in X_1$ that $\bar{\vp}(z)=e^{(s/D_a)z}$. By substituting
$\bar{\vp}(z)$  into \eqref{eqkerb} and \eqref{eqkerc} respectively and combining their boundary conditions \eqref{eqkerd}, \eqref{eqkere}, we obtain the expression of $\bar{\psi}(z)$ and $\bar{\phi}(z)$ in \eqref{eqbiaodashi}. This means that $$\dim\ker H=1~~\mbox{and}~~\ker H=\spa\{(\bar{\vp}(z),\bar{\psi}(z),\bar{\phi}(z))\}.
$$

Next we prove that $\cod\ran H=1$. If $(\xi_1(z),\xi_2(z),\xi_3(z),\xi_4,\xi_5)^T\in\ran H$, then there exists $(\hat{\vp}(z),\hat{\psi}(z),\hat{\phi}(z))\in X_1\times X_2\times X_3$ such that
\begin{equation}\label{eqran}
\begin{split}
&D_a\hat{\vp}''(z)-s\hat{\vp}'(z)=\xi_1(z),~0<z<L,\\
&D_n\hat{\psi}''(z)+((p_n-1)m_*-l/Q)\hat{\vp}(z)=\xi_2(z),~0<z<L,\\
&D_c\hat{\phi}''(z)-(m_*+(1-p_c)l/Q)\hat{\vp}(z)=\xi_3(z),~0<z<L,\\
&D_n\hat{\psi}'(L)+\al\hat{\psi}(L)=\xi_4,\\
&D_c\hat{\phi}'(0)-\bt\hat{\phi}(0)=\xi_5.
\end{split}
\end{equation}
Multiplying both sides of \eqref{eqkera} and the first equation of \eqref{eqran} by
$\hat{\vp}e^{-(s/D_a)z}$ and $\bar{\vp}e^{-(s/D_a)z}$, respectively, subtracting
and integrating on $[0,L]$, we have $\int_0^L\xi_1(z)\bar{\vp}(z)e^{-(s/D_a)z}dz=0$. Then
$$
\ran H=\left\{(\xi_1(z),\xi_2(z),\xi_3(z),\xi_4,\xi_5)\in Y^3\times\R^2:\int_0^L\xi_1(z)\bar{\vp}(z)e^{-(s/D_a)z}dz=0\right\}.
$$
and $\cod\ran H=1$.

We next prove that $F_{m,(U,N,C)}(m_*,0,N^0,C^0)(\bar{\vp},\bar{\psi},\bar{\phi})\not\in\ran H$. A direct calculation gives
$$
F_{m,(U,N,C)}(m_*,0,N^0,C^0)[\bar{\varphi}(z),\bar{\psi}(z),\bar{\phi}(z)]
=(-\bar{\vp}(z),c_np_n\bar{\vp}(z),0,0,0)^T.
$$
This shows that $F_{m,(U,N,C)}(m_*,0,N^0,C^0)(\bar{\vp},\bar{\psi},\bar{\phi})\not\in\ran H$ since $\int_0^L\bar{\vp}^2(z)e^{-(s/D_a)z}dz\neq0$.

By applying the Crandall-Rabinowitz bifurcation theorem (see \cite[Theorem 1.7]{Crandall1971}), we conclude that there exists a positive constant $\delta>0$ such that all
the coexistence steady state solutions of \eqref{fullequation} lie on a smooth curve
$\Gamma_1=\{(m(\tau),A(\tau,z),N(\tau,z),C(\tau,z)):0<\tau<\da\}$ satisfying \eqref{eqqncbiaodashiaaa}. In addition, we have
\begin{equation}\label{eqmdaoshu}
\begin{split}
m'(0)=&-\f{\left\langle \hat{l},F_{(A,N,C)(A,N,C)}\left(m_*, 0, N^0,C^0\right)[\vp(z),\psi(z),\phi(z)]^2\right\rangle}
{2\left\langle \hat{l}, F_{m,(A,N,C)}\left(m_*, 0, N^0,C^0\right)[\vp(z),\psi(z),\phi(z)]\right\rangle}\\
=&-\f{\il_0^L\left(\f{r\ga_n g(C^0)}{(\ga_n+N^0)^2}\bar{\psi}(z)
+\f{r\ga_c f(N^0)}{(\ga_c+C^0)^2}\bar{\phi}(z)\right)\bar{\vp}(z)e^{-(s/D_a)z}dz}
{-\il_0^L\bar{\vp}^2(z)e^{-(s/D_a)z}dz},
\end{split}
\end{equation}
where $\hat{l}$ is a linear functional on $ Y^3\times\R^2$ defined as
\[
\langle \hat{l},(\xi_1(z),\xi_2(z),\xi_3(z),\xi_4,\xi_5)\rangle=
\int_0^L\xi_1(z)\bar{\vp}(z)e^{-(s/D_a)z}dz.
\]
It follows from \eqref{eqbiaodashi} that
\begin{align*}
\bar{\psi}'(z)&=\f{D_a(l/Q+m_*(1-p_n))}{D_n s}e^{(s/D_a)z}-\f{D_a(l/Q+m_*(1-p_n))}{D_n s}\geq0,\\
\bar{\phi}'(z)&=\f{D_am_*}{D_cs}e^{(s/D_a)z}-\f{D_a(m_*+(1-p_c)l/Q)}{D_cs}e^{(s/D_a)L}\leq0
\end{align*}
for $z\in[0,L]$. This implies that $\bar{\psi}(z)$ is
nondecreasing and $\bar{\phi}(z)$ is nonincreasing in $z$.
It is clear that $\bar{\psi}(L)<0$ and $\bar{\phi}(0)<0$. Then $\bar{\psi}(z)<0$ and $\bar{\phi}(z)<0$ for any $z\in[0,L]$. From \eqref{eqmdaoshu} and $\bar{\vp}(z)=e^{(s/D_a)z}$, we have $m'(0)<0$.
This completes the proof of (i).

 (ii) In this installment, we will prove the uniqueness of bifurcation value $m_*$. Assume
that $\hat{m}$ is another bifurcation value from $\Gamma_0$, then there exists a positive solutions sequence $\{(m_n,A_n,N_n,C_n)\}$ of \eqref{steadystateequation} such that
$\{(m_n,A_n,N_n,C_n)\}\ra (\hat{m},0,N^0,C^0)$ in $C([0,L])$ as $n\ra\iy$. Let $\kappa_n=A_n/\|A_n\|_{\iy}$. From \eqref{steadystateequation}, $\kappa_n$ satisfies
\begin{equation*}
\begin{cases}
D_a\kappa_n''(z)-s\kappa_n'(z)+rf(N_n)g(C_n)\kappa_n(z)-(m_n+(l/Q))\kappa_n(z)=0,&0<z<L,\\
D_a\kappa_n'(0)-s\kappa_n(0)=D_a\kappa_n'(L)-s\kappa_n(L)=0.&
\end{cases}
\end{equation*}
It follows from $0<f(N_n)g(C_n)<1$ for all $z\in[0,L]$ that
$f(N_n)g(C_n)\ra f(N^0)g(C^0)$ in $C([0,L_1])$
as $n\ra\iy$. Note that $\{\kappa_n\},\{m_n\}$ are both bounded in
$L^\iy[0,L]$. By using $L^p$ theory for elliptic operators and
the Sobolev embedding theorem, there exists a subsequence of $\{\kappa_n\}$, denoted
by itself, and a function $\zeta\in C^2([0,L])$ such that
$\kappa_n\ra \zeta$ in $C^1([0,L_1])$ as $n\ra\iy$, and $\zeta$ satisfies (in the weak sense)
\begin{equation}\label{eqzeta}
\begin{cases}
D_a\zeta''(z)-s\zeta'(z)+rf(N^0)g(C^0)\zeta(z)-(\hat{m}+(l/Q))\zeta(z)=0,&0<z<L,\\
D_a\zeta'(0)-s\zeta(0)=D_a\zeta'(L)-s\zeta(L)=0.&
\end{cases}
\end{equation}
It follows from the strong
maximum principle that $\zeta>0$ on $[0,L]$ since
$\zeta\geq0$ and $\|\zeta\|_\iy=1$. From \eqref{mpoint} and \eqref{eqzeta}, we have
\begin{equation}\label{eqzetabbb}
\begin{cases}
D_a\zeta''(z)-s\zeta'(z)+(m_*-\hat{m})  \zeta(z)=0,&0<z<L,\\
D_a\zeta'(0)-s\zeta(0)=D_a\zeta'(L)-s\zeta(L)=0.&
\end{cases}
\end{equation}
Integrating from $0$ to $L$ on the both sides of the first equation of \eqref{eqzetabbb}, we have
$$
(m_*-\hat{m})\il_0^L\zeta(z)dz=0,
$$
which means that $m_*=\hat{m}$.

(iii) We now consider the local stability of $E_2(\tau)$. By using Corollary 1.13 and Theorem 1.16 in \cite{Crandall1973},  there exist continuously
differentiable functions $$\og_1: [m_*,
m_*+\da)\to
\R,~[\vp_1,\psi_1,\phi_1]: [m_*, m_*+\da)\to X_1\times X_2\times
 X_3,~\og_2: [0, \delta)\to \R$$ and
$[\vp_2,\psi_2,\phi_2]: [0, \delta)\to X_1\times
X_2\times X_3$ such that
\begin{align*}
&F_{(U,R,W)}\left(m, 0, N^0,C^0\right)[\vp_1^m(z),\psi_1^m(z),\phi_1^m(z)]
=\og_1(m)[\vp_1^m(z),\psi_1^m(z),\phi_1^m(z),0]^T,\\
&F_{(U,R,W)}\left(m(\tau), A(\tau,z), N(\tau,z), C(\tau,z)\right)[\vp_2^\tau(z),\psi_2^\tau(z),\phi_2^\tau(z)]
=\og_2(\tau)[\vp_2^\tau(z),\psi_2^\tau(z),\phi_2^\tau(z),0]^T
\end{align*}
for $z\in[0,L]$ and
\begin{equation}\label{eqstabilityog}
\lim\limits_{\tau\ra0}\f{-\tau m'(\tau)\og_1'(m_*)}{\og_2(\tau)}=1,
\end{equation}
where
$$\og_1(m_*)=\og_2(0)=0,~ (\vp_1^{m_*}(z),\psi_1^{m_*}(z),\phi_1^{m_*}(z))
=(\vp_2^0(z),\psi_2^0(z),\phi_2^0(z))=(\bar{\vp}(z),\bar{\psi}(z),\bar{\phi}(z)).
$$

It follows from \eqref{eqFrechetunc} that $\og_1(m)=m_*-m$, and
then $\og_1'(m_*)=-1$. From \eqref{eqzhubajieaa} and proof of (i), $\og_1(m_*)=0$
is the principal eigenvalue of $F_{(U,R,W)}\left(m_*, 0, N^0,C^0\right)$. By the perturbation theory of linear operators \cite{Kato1966}, we know that $\og_2(\tau)$ is also the principal eigenvalue of $F_{(U,R,W)}\left(m(\tau), A(\tau,z), N(\tau,z), C(\tau,z)\right)$ when $\tau$ is sufficiently small. Combining with $m'(0)<0$ and \eqref{eqstabilityog} gives $\og_2(\tau)<0$ for $\tau\in[0,\da)$. This means that $E_2(\tau)$ is locally asymptotically stable with respect to \eqref{fullequation}.
\end{proof}

We now consider the global existence of $E_2$ bifurcating from $\Gamma_0$ at $m=m_*$  with the help of the global bifurcation theory (see Theorem 3.3 and Remark 3.4 in \cite{ShiWang2009}).
First we establish the following {\it A priori} estimates  for
nonnegative solutions of \eqref{steadystateequation}.

\begin{lemma}\label{lepea}
Assume that $(A(z),N(z),C(z))$ is a nonnegative solution of \eqref{steadystateequation} with $A,N,C\not\equiv0$.
Then
\begin{itemize}

\item[\tu{(i)}] $0<m<m^*:=r-l/Q$;

\item[\tu{(ii)}] $0<N(z)<h_1, 0<C(z)<h_2$ for any $z\in[0,L]$, where
\begin{equation*}
    h_1=N^0+\frac{\al (r+p_nm^*)L Q N^0}{D_nl}, \;\; h_2=C^0+\frac{\al(rQ+p_cl)LQN^0}{D_cc_nlQ},
\end{equation*}

\item[\tu{(iii)}] There exists a
positive constant $B$ such that $\|A\|_\iy\leq B$ if
$m\in(0,m^*)$.
\end{itemize}
\end{lemma}

\begin{proof}
(i)  Let $U=e^{-(s/D_a)z}A$. Then
\begin{align*}
&-D_aU''(z)-sU'(z)+\left(m+\f{l}{Q}\right)U=rf(N)g(C)U\geq0,~0<z<L,\\
&U'(0)=U'(L)=0.
\end{align*}
By the strong maximum principle, we get $U>0$ on $[0,L]$, and then $A>0$ on $[0,L]$.
From \eqref{steadystateequation}, we have
\begin{equation}\label{eq231344}
\begin{cases}
-D_aA'(z)+sA'(z)-rf(N(z))g(C(z))A(z)=-\left(m+\f{l}{Q}\right)A(z),~0<z<L,\\
D_aA'(0)-sA(0)=D_aA(L)-sA(L)=0.
\end{cases}
\end{equation}
Hence the principal
eigenvalue of \eqref{eq231344} is
$\la_1\left(-rf(N(\cdot))g(C(\cdot))\right)=-(m+l/Q)$
with eigenfunction $A$. It follows from the monotonicity of the principal eigenvalue on the weight functions that
$$-r=\la_1
\left(-r\right)<\la_1
\left(-rf(N(\cdot))g(C(\cdot))\right)=-(m+l/Q).$$
This means that $0<m<r-l/Q=m^*$.

(ii) It follows from \eqref{steadystateequation} that
\begin{equation}\label{eqconganan}
\begin{split}
&\il_0^L\left(rf(N(z))g(C(z))-\left(m+\f{l}{Q}\right)\right)A(z)dz=0,\\
&\al(N^0-N(L))+\il_0^L c_n\left(p_nm-rf(N(z))g(C(z))\right)A(z)dz=0,\\
&-\bt(C(0)-C^0)+\il_0^L\left(\f{p_cl}{Q}-rf(N(z))g(C(z))\right)A(z)dz=0,
\end{split}
\end{equation}
which imply that $N(L)<N^0$ and $C(0)<C^0$. Note that
\begin{align*}
&-D_nN''(z)+\left(c_nrAg(C)\il_0^1f'(sN)ds\right)N=c_np_nmA\geq0,~0<z<L,\\
&-D_cC''(z)+\left(rAf(N)\il_0^1g'(sC)ds\right)C=\f{p_cl}{Q}A\geq0,~0<z<L,
\end{align*}
with the boundary conditions
$$
N'(0)=0,~D_nN'(L)=\al(N^0-N(L))>0,~-D_cC'(0)=\bt(C^0-C(0))>0,~C'(L)=0.
$$
From the strong maximum principle, we have $N>0$ and $C>0$ on $[0,L]$.

It follows from \eqref{eqconganan} that
\[
\il_0^LA(z)dz< \f{\al Q N^0}{c_nl}.
\]
For any $z\in[0,L]$, we obtain
\begin{align*}
|D_nN'(z)|&=\left|D_n\il_0^zN''(x)dx\right|=
\left|c_n\il_0^z\left(rf(N)g(C)-p_nm\right)A(x)dx\right|\\
&\leq \f{\al(r+p_nm^*) Q N^0}{l},
\end{align*}
and
\begin{align*}
|N(z)|&=|N(L)+N(z)-N(L)|\leq |N(L)|+|N(z)-N(L)|\\
&\leq N^0+\left| \il_z^LN'(x)dx \right|\leq N^0+\f{\al(r+p_nm^*) Q N^0}{D_nl}(L-z)\\
&< N^0+\f{\al (r+p_nm^*)L Q N^0}{D_nl}=h_1.
\end{align*}
On the other hand, for any $z\in[0,L]$, we get
\begin{align*}
|D_cC'(z)|&=\left|D_c\il_z^LC''(x)dx\right|=
\left|\il_z^L\left(rf(N)g(C)-\f{p_cl}{Q}\right)A(x)dx\right|\\
&\leq \f{\al(rQ+p_cl)QN^0}{c_nlQ},
\end{align*}
and
\begin{align*}
|C(z)|&=|C(0)+C(z)-C(0)|\leq |C(0)|+|C(z)-C(0)|\\
&\leq C^0+\left| \il_0^zC'(x)dx \right|\leq C^0+\f{\al(rQ+p_cl)QN^0}{D_cc_nlQ}(z-0)\\
&< C^0+\f{\al(rQ+p_cl)LQN^0}{D_cc_nlQ}=h_2.
\end{align*}

(iii) We now establish the boundedness of $A(z)$ for $m\in (0,m^*)$.
If it is not true, then we assume that
there are a sequence $m_i\in(0,m^*)$
and corresponding positive solutions $(A_i(z),N_i(z),C_i(z))$
of \eqref{steadystateequation} such that $\|A_i\|_\iy\ra\iy$ as
$i\ra\iy$. Without loss of generality, we may assume that
$m_{i}\ra m_0\in(0,m^*)$ as $i\ra\iy$.
From (ii), we obtain $0<f(N_i(z))g(C_i(z))<
f(h_1)g(h_2)$ for all $i$ and $z\in[0,L]$, and hence we also may assume
$f(N_i(z))g(C_i(z))\ra \zeta(z)$ weakly in $L^2(0,L)$ for some function
$\zeta(z)$ as $i\ra\iy$. Let $a_i=A_i/\|A_i\|_\iy$.
From the first equation of \eqref{steadystateequation}, we have
\begin{equation*}
\begin{cases}
-(D_aa'_i(z)-sa_i(z))'+\left(m_i+\f{l}{Q}\right)a_i(z)=rf(N_i(z))g(C_i(z))a_i(z),
~0<z<L,\\
D_aa'_i(0)-sa_i(0)=D_aa'_i(L)-sa_i(L)=0,\\
\il_0^{L}\left(rf(N_i(z))g(C_i(z))-\left(m_i+\f{l}{Q}\right)\right)a_i(z)dz=0.
\end{cases}
\end{equation*}
By using $L^p$ theory for elliptic operators and
the Sobolev embedding theorem,  we may assume (passing to a
subsequence if necessary) that $a_i\ra \xi$ in $C^1([0,L])$ as
$i\ra\iy$, and $\xi$ satisfies (in the weak sense)
\begin{equation}\label{eq234fsdfaaa}
\begin{cases}
-(D_a\xi'(z)-s\xi(z))'+\left(m_0+\f{l}{Q}\right)\xi(z)=r\zeta(z)\xi(z),~0<z<L,\\
D_a\xi'(0)-s\xi(0)=D_a\xi'(L)-s\xi(L)=0,\\
\il_0^{L}\left(r\zeta(z)-\left(m_0+\f{l}{Q}\right)\right)\xi(z)dz=0.
\end{cases}
\end{equation}
It follows from the strong
maximum principle that $\xi>0$ on $[0,L]$ since $\xi\geq0$ and $\|\xi\|_\iy=1$.
Note that
$N_i$ satisfies
\begin{equation}\label{eq2212334}
\begin{cases}
\f{D_nN_i^{''}(z)}{\|A_i\|_\iy}=c_u\left(rf(N_i(z))g(C_i(z))-p_nm_i\right)a_i(z),~0<z<L,\\
N'_i(0)=0,~D_nN'_i(L)=\al(N^0-N(L)).
\end{cases}
\end{equation}
Integrating \eqref{eq2212334} from $0$ to $L$, we have
$$
\f{\al(N^0-N(L))}{\|A_i\|_\iy}=c_n\il_0^L\left(rf(N_i(z))g(C_i(z))-p_nm_i\right)a_i(z)dz
$$
Letting $i\ra\iy$ gives
$$
0=c_n\il_0^L\left(r\zeta(z)-p_nm_0\right)\xi(z)dz=
c_n\il_0^L\left((1-p_n)m_0+\f{l}{Q}\right)\xi(z)dz>0,
$$
which is a contradiction. Hence (iii) holds.
\end{proof}

We have the following global bifurcation theorem of the branch of the coexistence steady state solution $E_2$.
\begin{theorem}\label{thancglobal}
If \eqref{mpoint} holds,  then there exists a connected component $\Upsilon^+$ of the set $\Upsilon$ of positive solutions to \eqref{steadystateequation} such that
\begin{itemize}
\item [\tu{(i)}] the closure of $\Upsilon^+$ includes the bifurcation point   $(m_*,0,N^0,C^0)$;
\item [\tu{(ii)}] the projection of $\Upsilon^+$ onto $\R^+$ via $(m,A_2(m,z),N_2(m,z),C_2(m,z))\ra m$ contains the interval $(0,m_*)$. In particular, \eqref{steadystateequation} has at least one positive solution for any $m\in (0,m_*)$.
\end{itemize}
\end{theorem}
\begin{proof}
It is easy to see that the conditions of Theorem 3.3 and Remark 3.4 in \cite{ShiWang2009}  hold by using standard ways (see \cite{ShiWang2009,ShiZhangZhang2019,WangXu2013}). This means that there exists a connected component $\Upsilon^+$ of
$\Upsilon$ containing $\Gamma_1=\{(m(\tau),A_2(\tau,z),N_2(\tau,z),C_2(\tau,z)):0<\tau<\da\}$. Moreover,  the closure of $\Upsilon^+$ includes the bifurcation point  $(m_*,0,N^0,C^0)$ and $\Upsilon^+$ satisfies one of the following three alternatives:
\begin{itemize}
\item [(1)] it is not compact in $\R\times X_1\times X_2\times X_3$;
\item [(2)] it contains another bifurcation point $(\hat{m},0,N^0,C^0)$, where $\hat{m}$ is another bifurcation value satisfying \eqref{eqker} with $\hat{m}\neq m_*$;
\item [(3)] it contains a point $(m,\hat{A}(z),N^0+\hat{N}(z),
    C^0+\hat{C}(z))$, where $0\neq(\hat{A}(z),\hat{N}(z),\hat{C}(z))\in Z$,
    $Z$ is a closed complement of $\ker L=\spa\{(\bar{\vp},\bar{\psi},\bar{\phi})\}$ in $X_1\times X_2\times X_3$.
\end{itemize}

It follows from the uniqueness of $m_*$ in Theorem \ref{theoremlocal} that case (2) cannot
happen. If case (3) holds, then
\begin{equation}\label{eqsunxue}
\il_0^{L}\hat{A}(z)\bar{\vp}(z)dz=0.
\end{equation}
From Lemma \ref{lepea}, we have $\hat{A}(z)>0$ on $[0,L]$, which is a contradiction to \eqref{eqsunxue}. Hence case
(1) must occur for $\Upsilon^+$.

From case (1), $\Upsilon^+$ is unbounded in $\R\times X_1\times X_2\times X_3$ .
By Lemma \ref{lepea}, if $m\in(0,m^*)$, then $(A_2(z),N_2(z),C_2(z))$ is bounded.
Therefore, the projection of $\Upsilon^+$ onto the $m$-axis contains the interval $(0,m_*)$ since $m_*<m^*$.
\end{proof}

Theorem \ref{thancglobal} shows that when $0<m<m_*$, \eqref{fullequation} always has a positive steady state solution $E_2$ while the nutrient-inorganic carbon-only semi-trivial steady state $E_1$ is unstable from Theorem \ref{thNClocal}. It is not known whether the positive steady state solution $E_2$ is unique although it appears to be from numerical simulations. In \cite{ZhangShiChang2018}, the uniqueness of positive steady state was shown for the case that algae only depends on nutrient.

\subsection{Simulations}

In this subsection, some numerical simulations are shown to illustrate
our analysis of steady states for system \eqref{fullequation}. Here
the set of parameter values we use is mainly from \cite{Jager2014,Nie2016,Van2011,Vasconcelos2016}. The values of all biologically reasonable parameters are listed in Table \ref{tabletparametersvalue}.

\begin{table}[th]
{ \begin{center} \caption[]{ Numerical values of
parameters of model \eqref{fullequation} with references.
}\label{tabletparametersvalue}\scriptsize
\begin{tabular}{p{1.2cm}p{1.4cm}p{1.6cm}p{1.9cm}p{1.2cm}p{1.4cm}p{1.6cm}p{1.9cm}}\hline
{\bf Symbol} &{\bf Values}&{\bf  Units}&{\bf
Source}&{\bf Symbol}
&{\bf Values}&{\bf  Units}&{\bf Source}\\
\hline $D_a,D_n,D_c$&0.02 (0.001-10)& m$^2$/day&
\cite{Huismanv2002,Huismanv2006,Jager2010,Klausmeier2001,Ryabov2010}&
$s$& 0.02(-0.5-0.5) &m/day&\cite{Huismanv2002,Huismanv2006,Jager2014,Jager2010,Klausmeier2001,Ryabov2010,Vasconcelos2016}\\
$r$&1.5 &day$^{-1}$&\cite{Vasconcelos2016}&
$m$& 0.1&day$^{-1}$&\cite{Jager2014,Vasconcelos2016}\\
$l$&50$\times$10$^{-9}$ &$\mu$molC cell$^{-1}$/day &\cite{Nie2016}&
$Q$&2500$\times$10$^{-9}$ &$\mu$molC cell$^{-1}$&\cite{Nie2016}\\
$\ga_n$&3
&mgP/m$^3$&\cite{Jager2014,Vasconcelos2016}&$\ga_c$& 500
&$\mu$molC/m$^3$&\cite{Van2011}\\
$c_n$&0.008&mgP/$\mu$molC &\cite{Jager2014,Vasconcelos2016}&
$p_n$&0.05 (0-1)
&--&Assumption\\
$p_c$&0.05(0-1)
&--&Assumption&$\al$&0.05
&m/day&\cite{Jager2014,Vasconcelos2016}\\
$\bt$&0.264&m/day&\cite{Van2011}&$N^0$&50 (5-400)
& mgP/m$^3$&\cite{Jager2014}\\
$C^0$&6250 (500-15000)&$\mu$molC/m$^3$&\cite{Van2011}&$L$&4
&m&Assumption\\
\hline
\end{tabular}
\end{center}}
\end{table}

\begin{figure}[htb!]
\centerline{\includegraphics[scale=0.7]{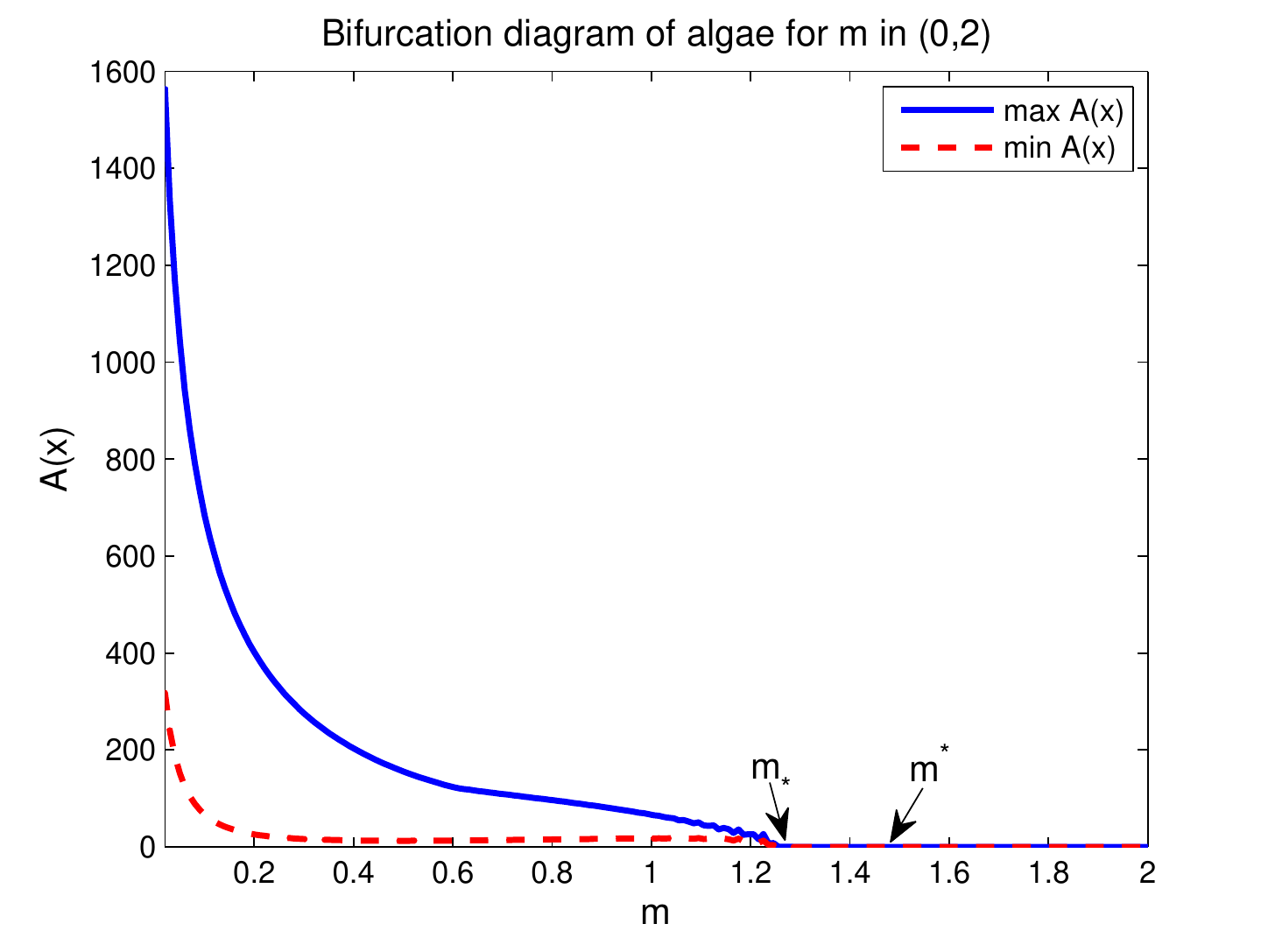}}
\caption{{\footnotesize Numerical bifurcation diagram of steady state algae density in \eqref{steadystateequation} for $m\in(0,2)$.  Here $D_a=D_n=D_c=0.1$ and other parameters except $m$ are from Table \ref{tabletparametersvalue}, and bifurcation values are $m^*=1.48$ and $m_*=1.29$.}}\label{parameterm}
\end{figure}

\begin{figure}[htb!]
\centerline{\includegraphics[scale=0.5]{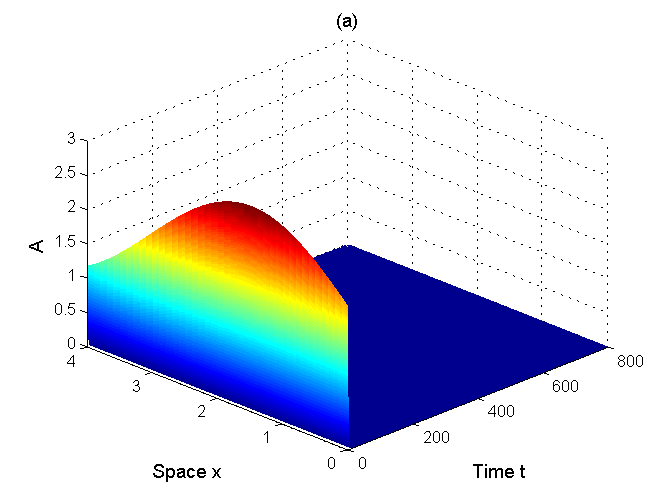} \includegraphics[scale=0.5]{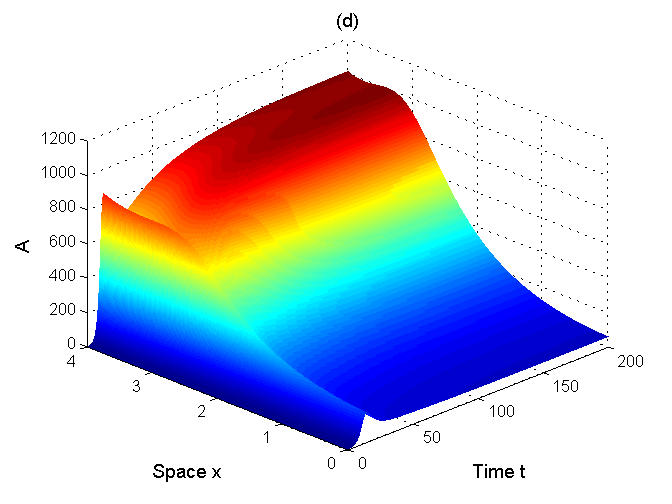}}
\centerline{\includegraphics[scale=0.5]{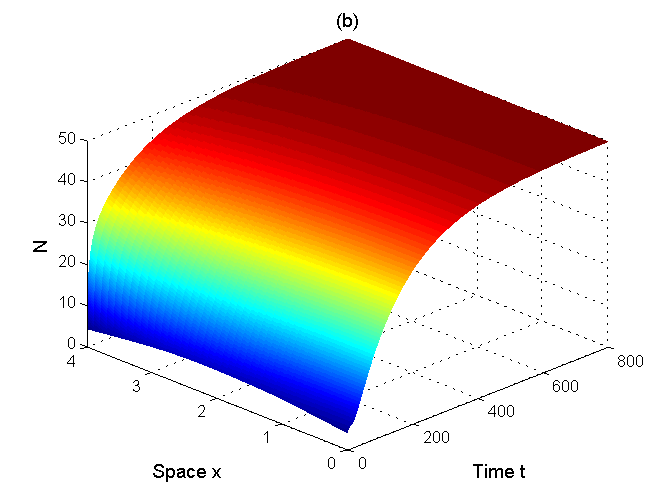} \includegraphics[scale=0.5]{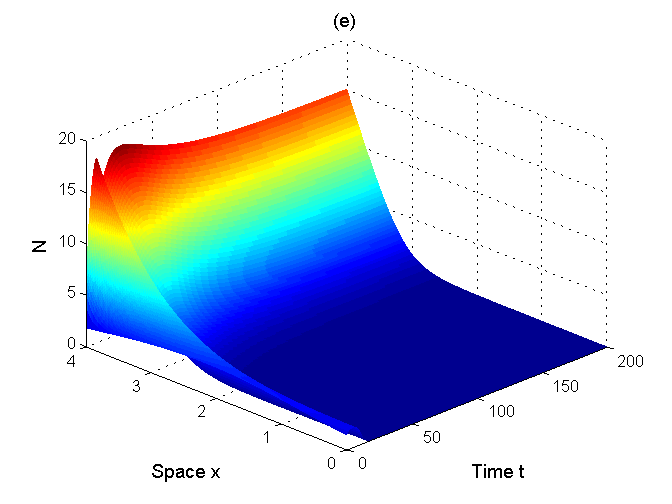}}
\centerline{\includegraphics[scale=0.5]{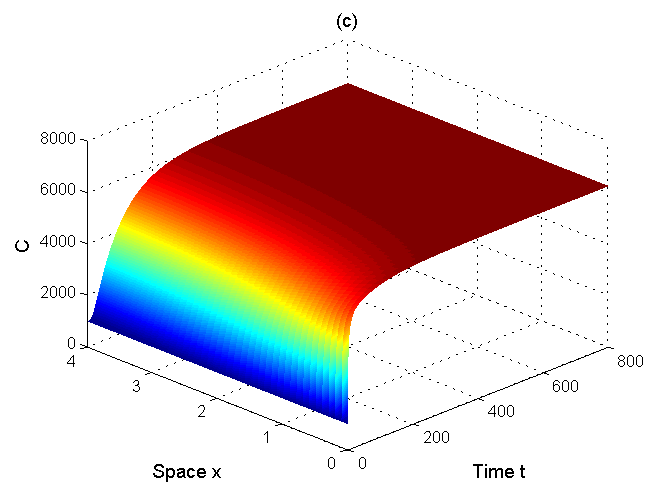} \includegraphics[scale=0.5]{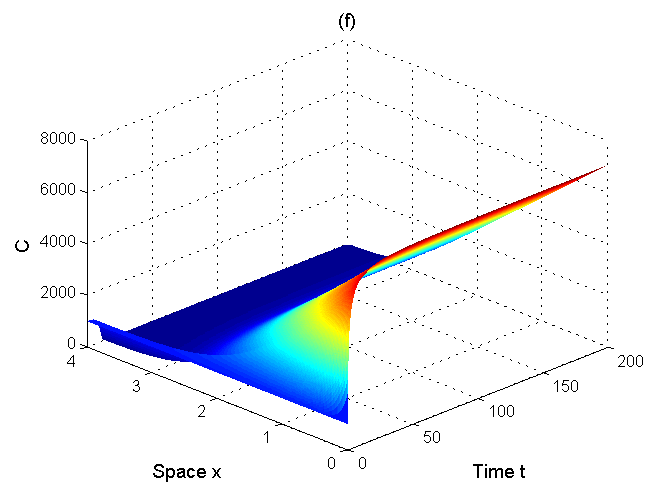}}
\caption{{\footnotesize  Left: the solution converges to the nutrients-inorganic carbon-only semi-trivial steady state $E_1$ in  (a)-(c) with $m=1.4$; Right: the solution converges to a coexistence steady state $E_2$ in Figure (d)-(f) with $m=0.1$. Here $D_a=D_n=D_c=0.1$ and other parameters except $m$ are from Table \ref{tabletparametersvalue}. }}\label{steaystate}
\end{figure}

According to Theorems \ref{thNClocal}, \ref{thNCglobal}, \ref{theoremlocal} and
\ref{thancglobal}, we can divide the parameter regime for $m$ according to two threshold values:
 $m_*=rf(N^0)g(C^0)-l/Q$ and $m^*=r-l/Q$.
$E_1$ always exists for all $m>0$, and  the stability of $E_1$ has the following three cases: (1) globally asymptotically stable when $m>m^*$; (2) locally asymptotically stable when $m_*<m<m^*$; (3) unstable when $0<m<m_*$. The bifurcation point $m=m_*$ is a threshold value  for the regime shift from extinction to survival of algae, such that the positive steady state solution $E_2$ exists when $m<m_*$ (see Figure \ref{parameterm}). The numerical bifurcation diagram in Figure \ref{parameterm} shows that the positive steady state $E_2$ decreases in $m$.

In Figure \ref{steaystate},  dynamic numerical  simulations of solutions of \eqref{fullequation} for parameter values from Table \ref{tabletparametersvalue} are shown.
For different algal loss
rates $m$, the solution of \eqref{fullequation} converges to different steady states regardless of initial conditions. 
For the case of $m=1.4>m_*$,  algae becomes extinct and the concentration of dissolved
nutrients and dissolved inorganic carbon reach the concentration of dissolved nutrients at the bottom of the water column and the concentration of dissolved inorganic carbon at the surface of the water column respectively (see Theorems \ref{thNClocal}, \ref{thNCglobal} and Figure \ref{steaystate} (a)-(c)).  For the case of $m=0.1<m_*$, algae, nutrients and inorganic carbon can coexist together at a positive level (see Theorems \ref{theoremlocal}, \ref{thancglobal} and Figure \ref{steaystate} (d)-(f)), and algae exhibit vertical aggregation phenomena (see Figure \ref{steaystate} (d)).

\section{The vertical distribution of algae}\label{sectionvertical}
\noindent

The vertical distribution of algae in the aquatic ecosystem is highly heterogeneous 
and affects the whole aquatic ecosystem in terms of primary productivity, trophic levels and cycle of matter. Especially, algae can exhibit the most prominent vertical aggregation phenomena, where the equivalent of 90\% of algal biomass sometimes gathers in a relatively thin layer \cite{Klausmeier2001,Ryabov2010,Yoshiyama2009,Yoshiyama2002}. In the present section, we will explore the influence of environmental parameters in \eqref{fullequation} on the vertical distribution of algae under the jointly limitation of nutrients and inorganic carbon when the solution of \eqref{fullequation} appears to converge to a coexistence steady state $E_2$.
In figures below, we compare the variation of the vertical distribution of coexistence steady state $A(z)$ for different parameter values.

\begin{figure}[htb!]
\centerline{\includegraphics[scale=0.5]{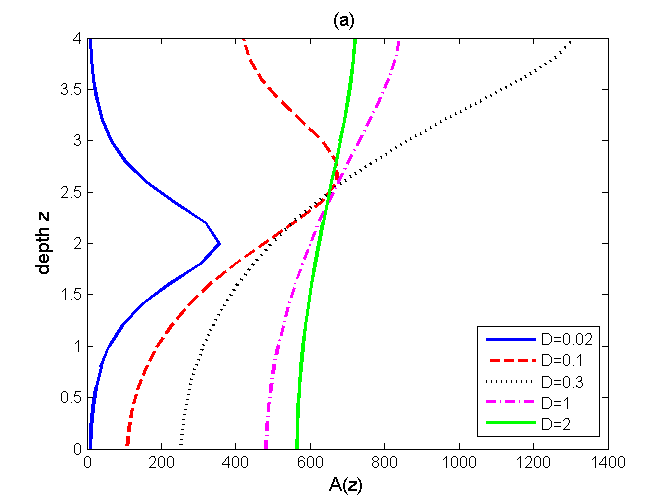}
\includegraphics[scale=0.5]{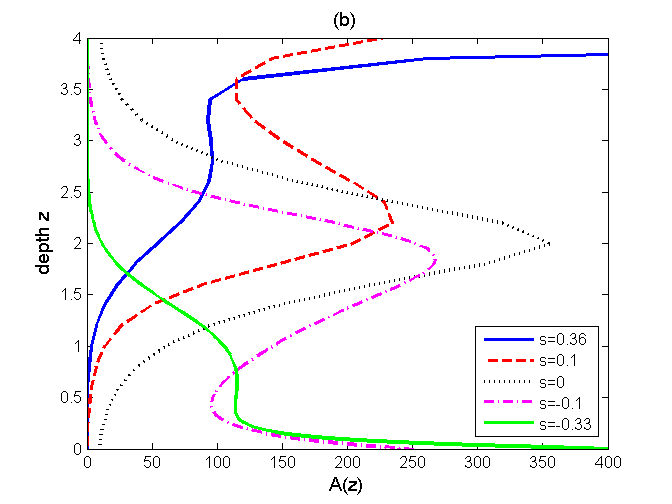}}
\caption{{\footnotesize  Vertical distribution of the coexistence steady state
$A(z)$ for varying $D=D_a=D_n=D_c$ (left) and $s$ (right). Here other parameters are from Table \ref{tabletparametersvalue}. }}\label{distribution-ds}
\end{figure}

We first explore the effect of diffusion coefficients $D_a,D_n,D_c$ and advection rate $s$ on the vertical distribution of algae. It has been proved that vertical turbulent diffusivity in lake and oceans  has obvious difference when the season changes \cite{Wuest2003}. It is generally true that turbulent diffusivity is small in summer, but large in winter. In order to facilitate the discussion, we let $D_a=D_n=D_c=D$. From Figure \ref{distribution-ds} (a), one can see that the spatial heterogeneity of algal biomass gradually changes from aggregation to even distribution when $D$ increases. This is because the increase of turbulent diffusivity leads to the full transfer of nutrients and organic carbon in the water column, so that every part of the water column is well adapted to the growth of algae. This confirms once again that algae are prone to aggregation in summer and not in winter because of turbulent diffusion. Therefore, to study the vertical distribution of algae better, we will consider only a poorly mixed water column and let $D_a=D_n=D_c=0.02$ in the remaining part of this section.

Suspension ($s=0$), subsidence ($s>0$) and floating upward ($s<0$) of algae are not only important ways to obtain the best growth opportunities, but also greatly affect the vertical accumulation of algae. From  Figure \ref{distribution-ds} (b), one can observe that the vertical distribution of algae changes greatly with the decrease of $s$ in a poorly mixed water column. For the case of $s=0.36$ and $s=-0.33$, the algal biomass is mainly concentrated at the bottom or the surface of the water column due to the larger sinking or buoyant velocity. When $s=0$, the maximum of the algal biomass, also described as deep chlorophyll maxima (DCMs), arises from the middle of the water column. This is because the nutrients are concentrated at the bottom, while inorganic carbon is concentrated at the surface for a lower
turbulent diffusion,  which leads to the optimal growth of algae in the middle of the water column. This indicates that the shift from subsidence to floating upward can cause algal blooms and bring out a large amount of algae gathering on the water surface.

\begin{figure}[b!]
\centerline{\includegraphics[scale=0.5]{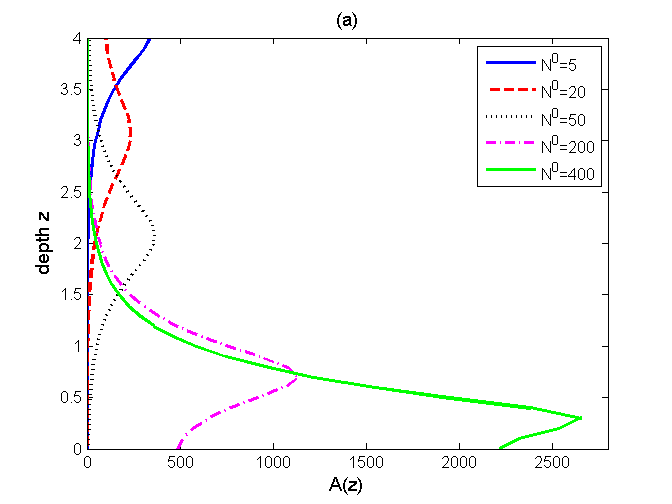}
\includegraphics[scale=0.5]{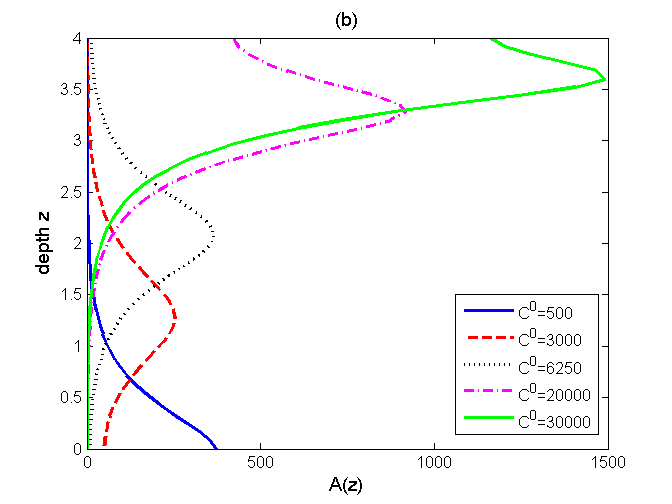}}
\caption{{\footnotesize  Vertical distribution of the coexistence steady state
$A(z)$ for varying $N^0$ (left) and $C^0$ (right). Here other parameters are from Table \ref{tabletparametersvalue}. }}\label{distribution-nc}
\end{figure}

Here an interesting phenomenon is the stratification of algae for the smaller sinking or buoyant velocity in Figure \ref{distribution-ds} (b). For $s=0.1$, the algal biomass is concentrated in the middle and bottom of the water column and are mainly divided into two layers: 2m-3m and 3.7m-4m. Similar phenomena also occurs for $s=-0.1$, where the algal biomass is concentrated in the middle and surface of the water column. In both cases, there are two local maximums of algae biomass density. This phenomenon shows that there may be one, two or even more layers for  the optimal growth of algae in a water column.

In our model \eqref{fullequation}, the algal growth is limited by organic carbon and nutrients. Nutrients from the water bottom and inorganic carbon from the water surface form an asymmetric resource supply mechanism on the algal growth. An increasing concentration $N^0$ of dissolved nutrients at the bottom reduces the dependence of the algal growth on nutrients in the whole water column, such that the algal biomass gradually increases and the aggregation location shifts from the bottom to the surface (see Figure \ref{distribution-nc} (a)). When the nutrient level on the water surface is very high, algae accumulate on the water surface and algal blooms occur (see $N^0=400$ in Figure \ref{distribution-nc} (a)). Hence eutrophication of water body caused by warm conditions, industrial waste water and sewage is an important factor of algal blooms. On the other hand, increasing organic carbon causes the algal biomass vertical aggregation transfering from the surface to the bottom (see Figure \ref{distribution-nc} (b)). This implies that the input of organic carbon can change the vertical distribution of algae, thus reduce the risk of algal blooms.

\begin{figure}[htb!]
\centerline{\includegraphics[scale=0.5]{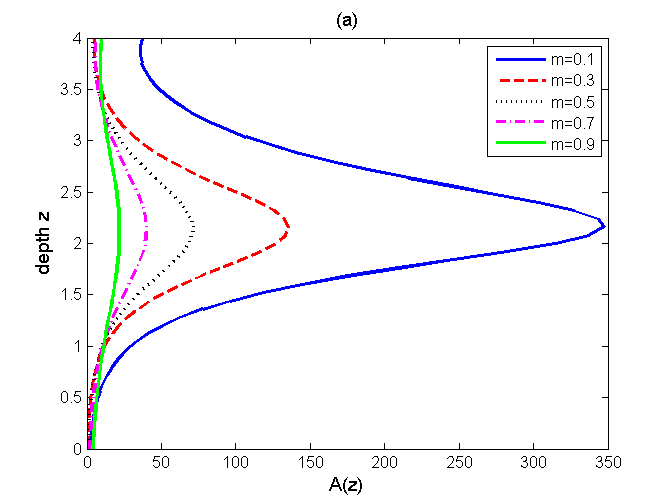}
\includegraphics[scale=0.5]{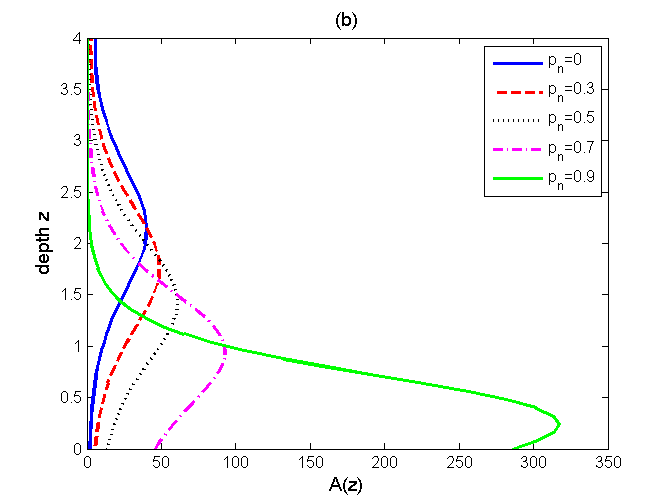}}
\caption{{\footnotesize  Vertical distribution of the coexistence steady state
$A(z)$ for varying $m$ (left) and $p_n$ (right). Here $p_n=0$ in (a), $m=0.7$ in (b)  and other parameters are from Table \ref{tabletparametersvalue}. }}\label{distribution-mpn}
\end{figure}

\begin{figure}[htb!]
\centerline{\includegraphics[scale=0.5]{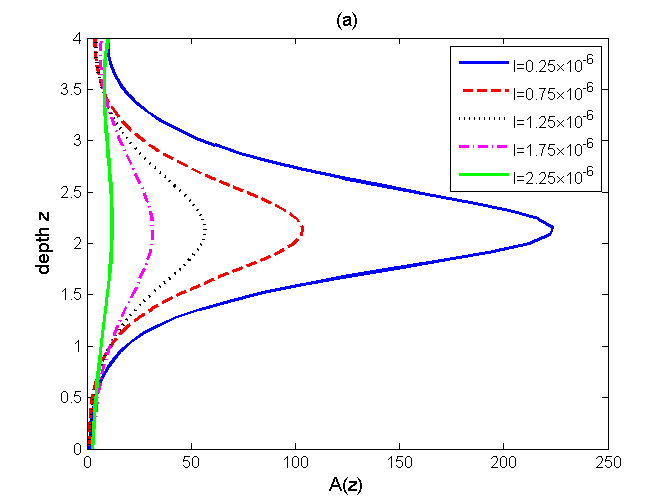}
\includegraphics[scale=0.5]{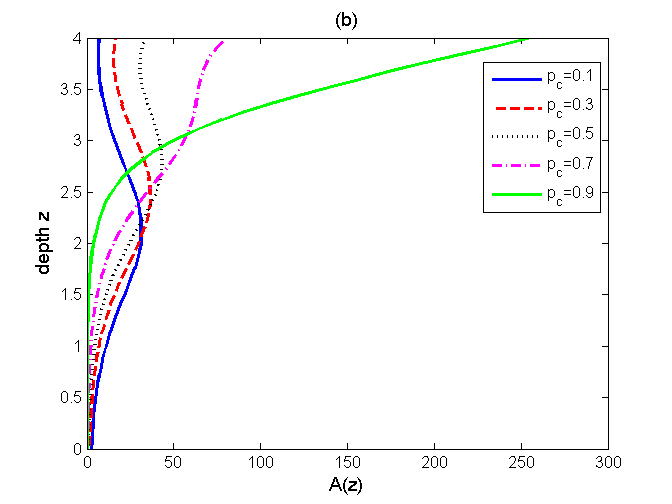}}
\caption{{\footnotesize  Vertical distribution of the coexistence steady state
$A(z)$ for varying $l$ (left) and $p_c$ (right). Here $p_c=0$ in (a), $l=1.75\times10^{-6}$ in (b) and other parameters are from Table \ref{tabletparametersvalue}. }}\label{distribution-lpc}
\end{figure}

The loss rate $m$ of algae only changes the total biomass of algae, but has no significant impact on the vertical distribution of algae (see Figure \ref{distribution-mpn} (a)).  Similarly, the respiration rate $l$ of algae also has no essential effect on the vertical distribution of algae (see Figure \ref{distribution-lpc} (a)).  From Figure \ref{distribution-mpn} (b) and Figure \ref{distribution-lpc} (b), one can observe that the nutrient recycling proportion $p_n$  and the organic carbon recycling proportion $p_c$ can both affect the vertical distribution of algae. As is often observed, a large amount of algae float on the surface during the daytime and sink to the bottom at night. The most important reason is the phototaxis of algae. Our results suggest that the high nutrient recycling proportion from loss of algal biomass is also a factor for the algal concentration on the surface during the day (see $p_n=0.9$ in Figure \ref{distribution-mpn} (b)). This happens because rapid decomposition of dead algae under high daytime temperature produces enough nutrients. In contrast, at night, the algal respiration releases a large amount of organic carbon, so that algae gather at the bottom (see $p_c=0.9$ in Figure \ref{distribution-lpc} (b)).

 Let $z_{\max}$ be the depth coordinate of the maximum of the algal biomass in the water column when the vertical aggregation occurs. In view of Figures \ref{distribution-ds}-\ref{distribution-lpc}, $z_{\max}$ describes the change of algal aggregation layer. The increase of $z_{\max}$ indicates that algae aggregate to the deep layer, while the decrease of $z_{max}$ indicates that algae aggregate to the surface layer. Let $A^*$ be the algal biomass on the water surface and $A_*$ be the algal biomass on the water bottom. Hence $A^*$ and $A_*$ are both important indices to measure algal blooms and the change of water quality. As a summary of the above discussion, the influence of environmental parameters
on $z_{\max}$, $A^*$ and $A_*$ are listed in Table \ref{tableparametersaggregation}.
Here we do not list the effects of turbulent diffusion. This is because with the increase of diffusion, the aggregation layer gradually disappears, and algae tend to be more evenly distributed.

\begin{table}[h]
{ \begin{center} \caption[]{ The influence of environmental parameters on
$z_{\max}$, $A^*$ and $A_*$.
}\label{tableparametersaggregation}
\begin{tabular}{p{2.8cm}p{2cm}p{2cm}p{2cm}}
\hline
{\bf Parameters} &{\bf $z_{\max}$}&{\bf  $A^*$}&{\bf
$A_*$}\\
\hline
$s$ $\uparrow$& $\uparrow$&$\downarrow$&$\uparrow$\\
$N^0$ $\uparrow$& $\downarrow$&$\uparrow$&$\downarrow$\\
$C^0$ $\uparrow$& $\uparrow$&$\downarrow$&$\uparrow$\\
$m$ $\uparrow$& ---&$\uparrow$&$\uparrow$\\
$p_n$ $\uparrow$& $\downarrow$&$\uparrow$&$\downarrow$\\
$l$ $\uparrow$& ---&$\uparrow$&$\uparrow$\\
$p_c$ $\uparrow$& $\uparrow$&$\downarrow$&$\uparrow$\\
\hline
\end{tabular}
$$\uparrow: \hbox{Increasing}~~~\downarrow:\hbox{
Decreasing}~~~\hbox{---: No significant effect}$$
\end{center}}
\end{table}

\section{Discussion}\label{sectiondiscussion}
\noindent

In this paper, we establish a reaction-diffusion-advection model \eqref{fullequation} to describe the dynamic interactions among algae, nutrients and inorganic carbon in a water column. The threshold conditions for the regime shift from extinction to existence of
algae is rigorously derived. By using theory of partial differential equations, dynamical systems, and  bifurcation theory, we establish theoretical results showing that algae beocme extinct when the loss rate $m$ exceeds the threshold value $m_*$ (see Theorems \ref{thNClocal} and \ref{thNCglobal}), while algae invade and persist when $m$ is below the threshold value $m_*$ (see Theorems \ref{theoremlocal} and \ref{thancglobal}).

The algal growth depends on limited nutrients from the water bottom and limited inorganic carbon from the water surface, which is an asymmetric supply of resources.
Our studies suggest that algae exhibit vertically spatial heterogeneity and vertical aggregation in the water column under the joint effect of nutrient and inorganic carbon supply (see Figures \ref{distribution-ds}-\ref{distribution-lpc}). Moreover, the environmental parameters of model \eqref{fullequation} could influence the algal vertical distribution (see Table \ref{tableparametersaggregation}).

From Figure \ref{distribution-ds} (a), one can see that vertical turbulent diffusion has an important effect on the vertically heterogeneous distribution of algae. If turbulent kinetic  energy is large enough, algae tend to distribute evenly throughout the water column, while small turbulent diffusion leads to aggregation of algae. The loss rate $m$ and the respiration rate $l$ only change the total biomass of algae, but have no effect on the location of algae aggregation layer (see Figures \ref{distribution-mpn} (a) and \ref{distribution-lpc} (a)). Parameters $s,N^0,C^0,p_n,p_c$ could affect the vertical distribution of algae, and are important indicators to measure the occurrence of algal blooms. Algal floating upward ($s<0$) causes algae to gather on the water surface, whereas the sinking of algae ($s>0$) causes algae to concentrate on the bottom (see Figure \ref{distribution-ds} (b)).  It is also observed that  there are algal vertical aggregation in multiple water layers for the intermediate advection rate (see $s=0.1, -0.1$  in Figure \ref{distribution-ds} (b)). In the nutrient-limited or inorganic carbon-limited water, the excessive input of nutrients is easy to cause algal blooms due to algal concentration on the water surface, while the input of organic carbon
reduces the risk of algal blooms as a result of algal aggregation on the water bottom (see Figure \ref{distribution-nc}). The recycling rates $p_n,p_c$ are both important indices to measure algal aggregation and blooms (see Figures \ref{distribution-mpn} (b) and \ref{distribution-lpc} (b)).

In previous studies, it has been confirmed that algae have vertical distribution and aggregation under the action of light and nutrients. In the present paper, our analysis shows that algae also exhibit complex spatial heterogeneity and vertical aggregation under the mechanism of asymmetric supply of nutrients and inorganic carbon. This study complements and further develops earlier studies of algae population growth in the water column. Based on the present discussion, it will be of interest to explore more biological questions. For example, two or more algae compete for nutrients and inorganic carbon;  the dynamics of algae population if all three resources (light, nutrients, inorganic carbon) are considered; the effect of toxic plankton species, zooplankton and fishes.

\bibliographystyle{plain}

\bibliography{algal-nutrients-carbin}

\end{document}